\documentclass[letterpaper, 11 pt,peerreviewca,onecolumn]{ieeeconf}
\usepackage{courier}
\usepackage{hyperref}
\usepackage{graphicx,subfigure}
\usepackage[scriptsize]{caption}
\usepackage{epstopdf}
\usepackage{tikz}
\usepackage{amsmath}
\usepackage{multirow}
\usepackage{amsfonts}
\usepackage{float}
\usepackage{amssymb}
\usepackage{graphicx}
\usepackage{flushend}
\usepackage[ruled]{algorithm}
\usepackage{algpseudocode}
\usepackage{amsthm}

\newtheorem{theorem}{Theorem}
\newtheorem{lemma}{Lemma}
\newtheorem{proposition}{Proposition}

\newtheorem{definition}{Definition}
\newtheorem{remark}{Remark}
\newtheorem{example}{Example}
\makeatletter

\newcommand{\Rmnum}[1]{\expandafter\@slowromancap\romannumeral #1@}

\makeatother

\begin{document}
\title{\LARGE Neuromimetic Control --- A Linear Model Paradigm}
\author{John Baillieul \& Zexin Sun}
\maketitle
\let\thefootnote\relax\footnotetext{\noindent
\hspace{-0.1in}\hrulefill
\hspace{0.8in}\\John Baillieul is with the Departments of Mechanical Engineering, Electrical and Computer Engineering, and the Division of Systems Engineering at Boston University, Boston, MA 02115. Zexin Sun is with the Division of Systems Engineering at Boston University.  The authors may be reached at {\tt \{johnb, zxsun\}@bu.edu}. \newline Support from various sources including the Office of Naval Research grants N00014-10-1-0952, N00014-17-1-2075, and N00014-19-1-2571 is gratefully acknowledged.\\
A condensed version of this paper has been submitted to the 60$^{th}${\em IEEE Conference on Decision and Control}.} 
\begin{abstract}
\noindent Stylized models of the neurodynamics that underpin sensory motor control in animals are proposed and studied.   The voluntary motions of animals are typically initiated by high level intentions created in the primary cortex through a combination of perceptions of the current state of the environment along with memories of past reactions to similar states.  Muscle movements are produced as a result of neural processes in which the parallel activity of large multiplicities of neurons generate signals that collectively lead to desired actions.  Essential to coordinated muscle movement are intentionality, prediction, regions of the cortex dealing with misperceptions of sensory cues, and a significant level of resilience with respect to disruptions in the neural pathways through which signals must propagate. While linear models of feedback control systems have been well studied over many decades, this paper proposes and analyzes a class of models whose aims are to capture some of the essential features of neural control of movement.  Whereas most linear models of feedback systems entail a state component whose dimension is higher than the number of inputs or outputs, the work that follows will treat models in which the numbers of input and output channels greatly exceed the state dimension.  While we begin by considering continuous-time systems governed by differential equations, the aim will be to treat systems whose evolution involves classes of inputs that emulate neural spike trains.  Within the proposed class of models, the paper will study resilience to channel dropouts, the ways in which noise and uncertainty can be mitigated by an appropriate notion of consensus among noisy inputs, and finally, by a simple model in which binary activations of a multiplicity of input channels produce a dynamic response that closely approximates the dynamics of a prescribed linear system whose inputs are continuous functions of time.
\end{abstract}
\begin{flushleft} {{\bf Keywords}: Neuromimetic control, parallel quantized actuation, channel intermittency, neural emulation} 
\end{flushleft}


 \section{Introduction}

The past two decades have seen dramatic increases in research devoted to information processing in networked control systems, \cite{Baillieul2007}.  At the same time, rapidly advancing technologies have spurred correspondingly expanded research activity in cyber-physical systems and autonomous machines ranging from assistive robots to self-driving automobiles, \cite{Leonard2008}.  Against this backdrop, a new focus of research that seeks connections between control systems and neuroscience has emerged, (\cite{Huang2018,Gawthrop,Khargonekar}), The aim is to understand how neurobiology should inform the design of control systems that can not only react in real-time to environmental stimuli but also display predictive and adaptive capabilities, \cite{Markkula}.  More fundamentally, however, the work described below seeks to understand how these capabilities might emerge from parallel activity of sets of simple inputs  that play a role analogous to sets of neurons in biological systems.  What follows are preliminary results treating linear control systems in which simple inputs steer the system by collective activity.  Principles of resilience with respect to channel dropouts and intermittency of channel availability are explored, and the stage is set for further study of prediction and learning.  The paper is organized as follows.  Section II introduces problems in designing {\em resilient} feedback stabilized motions of linear systems with (very) large numbers of inputs.  (Resilience here means that the designs still achieve design objective if a number of input channels become unavailable.)  Section III introduces an approach to resilient designs based on an approach we call {\em parameter lifting}.  Section IV briefly discusses how adding input channels which transmit noisy signals inputs can have the net effect of reducing uncertainty in achieving a control objective.  Section V takes up the problem of control with quantized inputs, and Section VI concludes with a discussion of ongoing work that is aimed at making further connections with problems in neurobiology.
\begin{figure}[h]
\begin{center}
\includegraphics[scale=0.25]{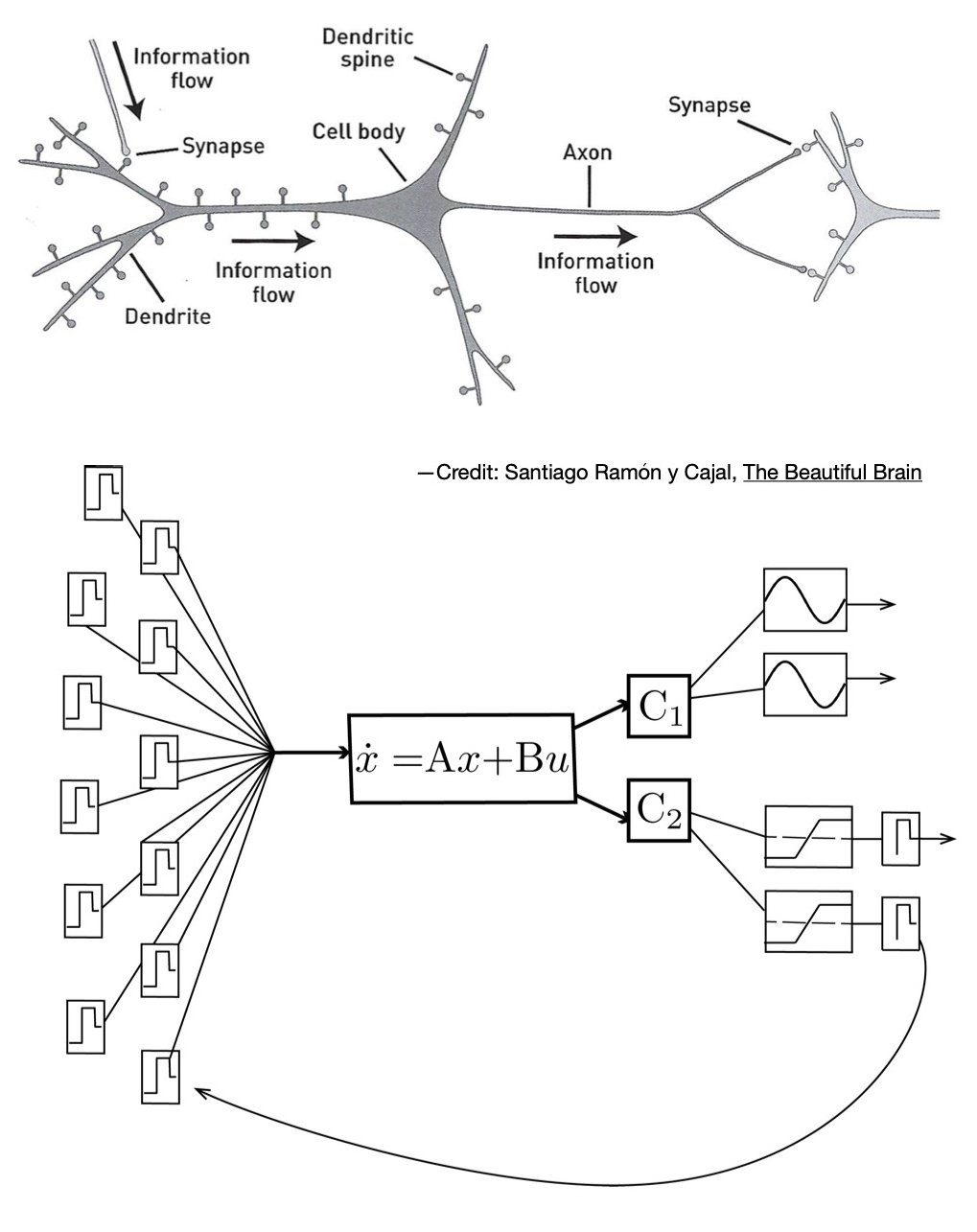}
\end{center}
\caption{The defining characteristics of neuromimetic linear systems are large numbers of input and output channels, all of which carry simple signals comprising possibly only a single bit of information and which collectively steer the system to achieve system goals.}
\end{figure}

\section{Linear Systems with Large Numbers of Inputs}

The models to be studied have the simple form
\begin{equation}
\begin{array}{l}
\dot x(t)=Ax(t) + Bu(t), \ \ \ x\in\mathbb{R}^n, \ \ u\in\mathbb{R}^m, \ {\rm and}\\[0.07in]
y(t)=Cx(t), \ \ \ \ \ \ \ \ \ \ \ \ \ \ y\in\mathbb{R}^q.\end{array}
\label{eq:jb:linear}
\end{equation}
As in \cite{Baillieul1,Baillieul2}. we shall be interested in the evolution and output of (\ref{eq:jb:linear}) in which a portion of the input or output channels may or may not be available over any given subinterval of time.  Among cases of interest, channels may intermittently switch in or out of operation.  In all cases, we are explicitly assuming that $m,q > n $.  In \cite{Baillieul1} we studied the way the structure of a system of this form might be affected by random unavailability of input channels.  The work in \cite{Baillieul2} showed the advantages of having large numbers of input channels (as measured by input energy costs as a function of the number of active input channels and resilience to channel drop-outs).  In what follows we shall show that having large numbers of parallel input channels provides corresponding advantages in dealing with noise and uncertainty.


To further explore the advantages of large numbers of inputs in terms of robustness to model and operating uncertainty and resilience with respect to input channel dropouts, we shall examine control laws for systems of the form (\ref{eq:jb:linear}) in which groups of control primitives are chosen from dictionaries and aggregated on the fly to achieve desired ends.  The ultimate goal is to understand how carefully aggregated finitely quantized inputs can be used to steer (\ref{eq:jb:linear}) as desired.  To lay the foundation for this inquiry, we begin by focusing on dictionaries of continuous control primitives that we shall call {\em set point stabilizing}.


\subsection{Resilient eigenvalue assignment for systems with large numbers of inputs}

Briefly introduced in \cite{Baillieul2} our dictionaries will be comprised of {\em set-point stabilizing} control primitives  of the form 
\[
u_j(t)=v_j+k_{j1}x_1 + \dots + k_{jn}x_n,\ \ j=1,\dots,m,
\]
where the gains $k_{ji}$ are chosen to make the matrix $A+BK$ Hurwitz, and the $v_j$'s are then chosen to make the desired goal point $x_g\in\mathbb{R}^n$ an equilibrium of (\ref{eq:jb:linear}).  Thus, given $x_g$ and a desired gain matrix $K$ the vector $v$ can chosen so as to satisfy the equation

\begin{equation}
(A+BK)x_g+Bv=0.
\label{OffSet}
\end{equation}

\begin{proposition}
Let $B$ be full rank $=n$, and let the $m\times n$ matrix $K$ be such that the eigenvalues of $A+BK$ are in the open left half plane. If $v$ is any vector satisfying (\ref{OffSet}), then the $m$ control inputs
\begin{equation}
u_j(t)=v_j + k_{j1}x_1(t) + \cdots + k_{jn}x_n(t)
\label{eq:jb:standard}
\end{equation}
steer (\ref{eq:jb:linear}) toward the goal $x_g$.
\end{proposition}

  Since $m>n$, the problem of finding values $k_{ij}$ making $A+BK$ Hurwitz and $v$ satisfying (\ref{OffSet}) is underdetermined.  We can thus carry out a parametric exploration of ranges of values $K$ and $v$ that make the system (\ref{eq:jb:linear}) {\em resilient} in the sense that the loss of one (or more) input channels will not prevent the system from reaching the desired goal point.
To examine the effect of channel unavailability and channel intermittency, let $P$ be an $m\times m$ diagonal matrix whose diagonal entries are $k$ 1's and $m$-$k$ 0's.  For each of the $2^m$ possible projection matrices of this form, we shall be interested in cases where $(A,BP)$ is a controllable pair.  We have the following:

\begin{definition}\rm
Let $P$ be such a projection matrix with $k$ $1's$ on the main diagonal.  The system (\ref{eq:jb:linear}) is said to be $k$-{\em channel controllable with respect to} $P$ if for all $T>0$, the matrix
\[
W_P(0,T)= \int_0^T\, e^{A(T-s)}BPB^T{e^{A(T-s)}}^T\, ds.
\]
is nonsingular.
\label{def:jb:One}
\end{definition}

\begin{remark} \rm For the system $\dot x=Ax+Bu$, being $k$-channel controllable with respect to a canonical projection $P$ having $k$ ones on the diagonal is equivalent to $(A,BP)$ being a controllable pair.
\end{remark}

\smallskip



\begin{example} \rm
Consider the three input system
\begin{equation}
\left(\begin{array}{c}
\dot x_1 \\
\dot x_2\end{array}\right)
=\left(\begin{array}{cc}
0 & 1\\
0 & 0\end{array}
\right)
\left(\begin{array}{c}
x_1 \\
x_2\end{array}\right)
+\left(
\begin{array}{ccc}
 0 & 1 & 1 \\
 1 & 0 & 1 \\
\end{array}
\right)u.
\label{eq:jb:kChannel}
\end{equation}
Adopting the notation
\[
P[i,j,k]=
\left(
\begin{array}{ccc}
 i & 0 & 0 \\
 0 & j & 0 \\
 0 & 0 &k
\end{array}\right),
\]
the system (\ref{eq:jb:kChannel}) is 3-channel controllable with respect to $P[1,1,1]$; it is 2-channel controllable with respect to $P[1,1,0],P[1,0,1],$ and $P[0,1,1]$.  It is 1-channel controllable with respect to $P[1,0,0]$ and $P[0,0,1]$, but it fails to be 1-channel controllable with respect to $P[0,1,0]$.  
\end{example}


Within classes of system (\ref{eq:jb:linear}) that are $k-$channel controllable, for $1\le k\le m$, we wish to characterize control designs that achieve set point goals despite different sets of $j$ control channels ($1\le j\le m-k$) being either intermittently or perhaps even entirely unavailable.  When $m>n$, the problem of finding $K$ and $v$ such that $A+BK$ is Hurwitz and (\ref{OffSet}) is satisfied leaves considerable room to explore solutions $(K,v)$ such that $A+BK$ is Hurwitz and  $(A+BPK)x_g + BPv=0$ for various coordinate projections of the type considered in Definition \ref{def:jb:One}.



To take a deeper dive into the theory of LTI systems with large numbers of input channels, we introduce some notation.  Fix integers $0 \le k \le m$, and let $[m] := \{1,\dots, m\}$.  Let $[m]\choose k$ be the set of $k$-element
subsets of [m].  In Example 3, for instance, ${[3]\choose 2} = \left\{\{1,2\},\{1,3\},\{2,3\}\right\}$.  Extending the $2\times 3$ example, we consider matrix pairs $(A,B)$ where $A$ is $n\times n$ and $B$ is $n\times m$.  We consider the following:
\bigskip

\noindent{\bf Problem A.}  Find an $m\times n$ gain matrix $K$ such that $A+BP_IK$ is Hurwitz for projections $P_I$ onto the coordinates corresponding to the index set $I$ for all $I\in {[m]\choose j}$ with $n\le j\le m$.

\medskip

\noindent{\bf Problem B.}  Find an $m\times n$ gain matrix $K$ that assigns eigenvalues of $A+BP_IK$ to specific values for each $I\in {[m]\choose n}$.

\bigskip

\noindent
Both of these problems can be solved if $m=n+1$, but when $m>n+1$, there are more equations than unknowns, making the problems over constrained and thus generally not solvable.  
We consider the following:

\medskip 

\begin{definition}\rm 
Given a system (\ref{eq:jb:linear}) where $A$ and $B$ are respectively $n\times n$ and $n\times m$ matrices with $n\le m$, for each $I\in{[m]\choose n}$, the $I$-th {\it principal subsystem} is given by the state evolution
\[
\dot x(t) = Ax(t)+BP_Iu(t).
\]
\begin{flushright}
$\square$
\end{flushright}
\end{definition}

Problem B requires solving $n$ equations for each $I\in{[m]\choose n}$, and if we seek an $m\times n$ gain matrix $K$ which places eigenvalues of $A+BP_IK$ for every$I\in{[m]\choose n}$, then a total of $n{m\choose n}$the simultaneous equations must be solved to determine the $n m$ entries in $K$.  Noting that $nm\le n{m\choose n}$, with the inequality being strict if $m>n+1$, we see that solving Problem B cannot be solved in general.  Problem A is less exacting in that it only requires eigenvalues to lie in the open left half plane, but it carries the further requirement that a single gain $K$ places all the closed-loop eigenvalues of all $I-th$ subsystems in the open left half plane for $I\in{[m]\choose n}$ and all $k$ in the range $n\le k\le m$.  The following example shows that solutions to Problem B are not necessarily solutions to Problem A as well and illustrates the complexity of the {\em resilient eigenvalue placement problem}.

\medskip 

\begin{example} \rm
For the system (\ref{eq:jb:kChannel}) of Example 1, we consider the three 2-channel controllable pairs
\[
(A,BP[110]),\ \ \  (A,BP[101]), \ \ \ (A,BP[011]).
\]
We look for a $3\times 2$ gain matrix $K$ such that $A+BPK$ has specified eigenvalues for each of these $P$'s.  Thus we seek to assign eigenvalues to the three matrices
\[
\left(\begin{array}{cc}
  k_{21} & 1 +k_{22}\\
  k_{11} & k_{12}\end{array}
  \right),\ \ \ 
 \left(\begin{array}{cc}
  k_{31} & 1 +k_{32}\\
  k_{11} +k_{31} & k_{12} +k_{32}\end{array}
  \right), 
  \]
  \[
 \left(\begin{array}{cc}
  k_{21} +k_{31} & 1 +k_{22} + k_{32}\\
  k_{31} & k_{32}\end{array}
  \right)
\]
respectively.  For any choice of LHP eigenvalues, this requires solving six equations in six unknowns.  For all three matrices to have eigenvalues at (-1,-1)  the $k_{ij}$-values $k_{11}=0 ,k_{12}= -1,k_{21}= -1,k_{22}= 0, k_{31}= -1/2,k_{32}= -1/2$ place these closed loop eigenvalues as desired.  For this choice of $K$, the closed loop eigenvalues of $(A + BK)$ are $-3/2\pm i/2$.  It is not possible to assign all eigenvalues by state feedback independently for the four controllable pairs $(A,B)$, $(A,BP[110])$, $(A,B,P[101])$, and $(A,BP[011])$.  Moreover, it is coincidental that the eigenvalues of $A+BK$ are in the left half plane.  This is seen in the following examples.  
(Consider $k_{11}=-2 ,k_{12}= -3,k_{21}= -1,k_{22}= -2.4, k_{31}= 0,k_{32}= -1/2$.  There are the following closed-loop eigenvalues: For $A+BP[110]K$: $(-3.95,-0.05)$, for $A+BP[101]K$: $(-3.19,-0.31)$, and for $A+BP[011]K$: $(-1,-0.5)$; but the eigenvalues of $A+BK$ are $(-4.57,0.066)$.)
\end{example}

Resilience requires an approach that applies simultaneously to all pairs ($A,BP_I$) where $P_I$ ranges over the lattice of projections $P_I$, $I\in {[m]\choose j},\ j=n\dots,m$.
In view of these examples, it is of interest to explore conditions under which the stability of feedback designs will be preserved as control input channels become intermittently or even permanently unavailable.


\section{Lifted Parameters}

 To exploit the advantages of a large number of control input channels, we turn our attention to using the extra degrees of freedom in specifying the control offsets $v_1,\dots,v_m$ so as to make ({\ref{eq:jb:linear}) resilient in the face of channels being intermittently unavailable.  As noted in \cite{Baillieul2}, we seek an offset vector $v$ that satisfies
$B\left[(\hat A + K)x_g +v\right]=0$
where $\hat A$ is any $m\times n$ matrix satisfying $B\hat A = A$.  If $K$ is chosen such that $A+BK$ is Hurwitz, the feedback law $u=Kx+v$ satisfies (\ref{OffSet}) and by Proposition 1 steers the state of (\ref{eq:jb:linear}) asymptotically toward the goal in $x_g$.  Under the assumption that $B$ has full rank $n$, such matrix solutions can be found--although such an $\hat A$ will not be unique.  Once $\hat A$ and the gain matrix $K$ have been chosen, the offset vector $v$ is determined by the equation 
\begin{equation}
(\hat A + K)x_g+v=0.
\label{OffSet2}
\end{equation}
Conditions under which $\hat A$ may be chosen to make (\ref{eq:jb:linear}) resilient to channel dropouts is the following.

\begin{theorem}
(\cite{Baillieul2}) Consider the linear system (\ref{eq:jb:linear}) in which the number of control inputs, $m$, is strictly larger than the dimension of the state, $n$ and in which rank $B=n$.  Let the gain $K$ be chosen such that $A+BK$ is Hurwitz, and assume that\\
(a)  $P$ is a projection of the form considered in Definition \ref{def:jb:One} and (\ref{eq:jb:linear}) is ${\ell}$-channel controllable with respect to $P$;\\
(b)  $A+BPK$ is Hurwitz;\\
(c) the solution $\hat A$ of $B\hat A=A$ is invariant under $P$---i.e., $P\hat A = \hat A$; and\\
(d) $BP$ has rank $n$.
Then the control inputs defined by (\ref{eq:jb:standard}) steer (\ref{eq:jb:linear}) toward the goal point $x_g$ whether or not the $(m-\ell)$ input channels that are mapped to zero by $P$ are available. 
\label{thm:jb:four}
\end{theorem}

The next two lemmas review simple facts about the input matrices $B$ under consideration.

\begin{lemma}\rm
Let $B$ be an $n\times m$ matrix with $m\ge n$.  If all principal minors of $B$ are nonzero, then $B$ has full rank.
\end{lemma}
\begin{proof}
The matrix $B$ has either more or the same number of columns as rows.  Its rank is thus the number of linearly independent rows.  If there is a nontrivial linear combination of rows of $B$, then this linear combination is inherited by all $n\times n$ submatrices, and thus all principal minors would be zero.
\end{proof}

\begin{lemma}\rm
Let $B$ be an $n\times m$ matrix with $m\ge n$ and having full rank $n$.  Then $BB^T$ is positive definite.
\label{lem:jb:pd}
\end{lemma}
\begin{proof}
If $B$ has full rank, the $n$ rows are linearly independent.  Hence if $x^T B$ is the zero vector in $\mathbb{R}^m$, then we must have $x=0$.  Thus, for all $x\in \mathbb{R}^n$, $x^TBB^Tx = \Vert B^Tx\Vert^2 \ge 0$ with equality holding if and only if $x=0$.
\end{proof}

\begin{lemma}\rm
Consider a linear function $B:\mathbb{R}^m\to\mathbb{R}^n$ given by an $n\times m$ rank $n$ matrix $B$ with $m\ge n$.  This induces a linear function $\hat {\rm B}$ having rank $n^2$ from the space of $m\times n$ matrices, $\mathbb{R}^{m\times n}$, to the space of $n\times n$ matrices, $\mathbb{R}^{n\times n}$, which is given explicitly by $\hat {\rm B}$(Y) = $B\cdot $Y.
\end{lemma}
\begin{proof}
Because $B$ has full rank $n$, it has a rank $n$ right inverse $U$.  Given any $n\times n$ matrix $A$, the image of $UA\in\mathbb{R}^{m\times n}$ under $\hat{\rm B}$ is $A$, proving the lemma.
\end{proof}



\smallskip
This function lifting is depicted graphically by 
\[
\begin{array}{clcc}
\hat {\rm B}:&\mathbb{R}^{m\times n}&\rightarrow&\mathbb{R}^{n\times n}\\
&\ \ \big\vert&&\big\vert\\
B:&\mathbb{R}^m&\rightarrow&\mathbb{R}^n.
\end{array}
\]

\begin{lemma}\rm 
Given that the rank of the $n\times m$ matrix $B$ is $n$ with $n < m$, the dimension of the null space of $B$ is $m-n$.  The dimension of the null space of $\hat {\rm B}$ is $n(m-n)$.
\label{lem:jb:null}
\end{lemma}
\begin{proof}
Let the set column vectors $\left\{\vec n_1,\dots,\vec n_{m-n}\right\}$ be a basis of the nullspace of $B$.  The set of all $m\times n$ matrices of the form $N=(\vec n_{i_1}\vdots\vec n_{i_2}\vdots\cdots\vdots\vec n_{i_n})$, where $i_j\in[m-n]$, is a linearly independent set that spans the null space of $\hat {\rm B}$.  There are $m-n$ independent choices of labels for each of the $n$ columns of $N$, proving the lemma.
\end{proof}

\begin{lemma}\rm
Let $A$ be an $n\times n$ matrix and $B$ be an $n\times m$ matrix with $m > n$ and $B$ having full rank $n$.  Then there is an $n(m-n)$-parameter family of solutions $X=\hat A$ to the matrix equation $\hat {\rm B}(X)=A$.
\end{lemma}
\begin{proof}
From Lemma \ref{lem:jb:null}, we may find a basis $\{N_1,\dots,N_{n(m-n)}\}$ of the nullspace of $\hat {\rm B}$, which we denote ${\cal N}(\hat {\rm B})$.  A particular solution of the matrix equation is $\hat A=B^T(BB^T)^{-1}A$, and any other solution may be written as $B^T(BB^T)^{-1}A+N$ where $N\in{\cal N}(\hat {\rm B})$.  This proves the lemma.
\end{proof}

\begin{lemma} \rm
Consider the LTI system $\dot x = Ax + Bu$, where $A$ is an $n\times n$ real matrix and $B$ is an $n\times m$ real matrix with $m>n$.  Suppose that all principal minors of $B$ are nonzero, and further assume that $A=0$.  Then, for any real number $\alpha>0$, the feedback gain $K=-\alpha B^T$ is stabilizing; i.e.\ $BK$ has all eigenvalues in the open left half plane.  Moreover, for all $n\le j\le m$ and $I\in {[m]\choose j}$, all matrices $BP_IK$ have eigenvalues in the open left half plane.
\label{lem:jb:canon}
\end{lemma}
\begin{proof}
Let $\bar B = P_IB$.  Noting that $\bar B$ has full rank $n$ and $\bar B\bar B^T=BP_IB^T$, the lemma follows from Lemma \ref{lem:jb:pd}.

\end{proof}

\begin{lemma} \rm
Consider the LTI system $\dot x = Ax + Bu$, where $A$ is an $n\times n$ real matrix and $B$ is an $n\times m$ real matrix with $m>n$.  Suppose that all principal minors of $B$ are nonzero, and further assume that $\hat A$ is a solution of $\hat {\rm B}(\hat A) = A$.  Then, for any real $\alpha>0$ and feedback gain 
\begin{equation}
K = -\alpha B^T-\hat A,
\label{eq:jb:defK}
\end{equation}
the closed loop system $\dot x=(A+BK)x$ is exponentially stable---i.e..the eigenvalues of $A+BK$ lie in the open left half plane.
\label{lem:jb:lastOfSec}
\end{lemma}
 \begin{proof}
Substituting $u=Kx$ into $\dot x=Ax + Bu$, we find that the closed-loop system is given by $\dot x = -\alpha BB^Tx$, and the conclusion follows from Lemma \ref{lem:jb:canon}.
 \end{proof}

\begin{definition}\rm
Let $I\in{[m]\choose j}$.  The {\it lattice of index supersets of I} in $[m]$ is given by ${\cal L}_I=\{L\subset [m]: I\subset L\}$.
\end{definition}

\begin{theorem} \rm
Consider the LTI system $\dot x = Ax + Bu$, where $A$ is an $n\times n$ real matrix and $B$ is an $n\times m$ real matrix with $m>n$.  Suppose that all principal minors of $B$ are nonzero, and further assume that $\hat A$ is a solution of $\hat {\rm B}(\hat A) = A$, and that for a given $I\in {[m]\choose j}$, $\hat A$ is invariant under $P_I$---i.e.$P_I\hat A=\hat A$.  Then, for any real $\alpha>0$ and feedback gain $K = -\alpha B^T-\hat A$ the closed loop systems $\dot x=(A+BP_LK)x$ are exponentially stable for all $L\in{\cal L}_I$.
\label{thm:jb:theorem2}
\end{theorem}
\begin{proof} The theorem follows from Lemma \ref{lem:jb:lastOfSec} with $BP_I$ and $P_IK$ here playing the roles of $B$ and $K$ in the lemma.
\end{proof}

Within the class of resilient (invariant under projections) stabilizing feedback designs identified in Theorem \ref{thm:jb:theorem2}, there is considerable freedom to vary paramters to meet system objectives.  To examine parameter choices, let $P_{[k]}$ be a diagonal $m\times m$ projection matrix having $k$-0's and $m-k$-1's on the principal diagonal ($1\le k\le m-n$).  For each such $P_{[k]}$, the $n(m-n)$ parameter family of solutions $\hat A$ to $\hat {\rm B}(\hat A)=A$ is further constrained by the invariance $P_{[k]}\hat A = \hat A$ which imposes $kn$ further constraints.  Hence, the family of $P_{[k]}$-invariant $\hat A$'s is $n(m-n-k)$-dimensional; see Fig.\ \ref{fig:zs:Lattice}.  Within this family, design choices may involve adjusting the overall strength of the feedback (parameter $\alpha$) or differentially adjusting the influence of each input channel by scaling rows of $B^T$ in (\ref{eq:jb:defK}).  Such weighting will be discussed in greater detail in Section \ref{sec:jb:quantized} where we note that to make the models reflective of the parallel actions of many simple inputs, it is necessary to normalize the weight $\alpha$ of $B^T$ in (\ref{eq:jb:defK}) so as to not let the influence of inputs depend on the total number of those available (as opposed to the total number acting according to the various projections $P$).  In other words, we want to take care not to overly weight the influence of large groups versus small groups of channels.


\begin{figure}[h]
\begin{center}
\includegraphics[scale=0.3]{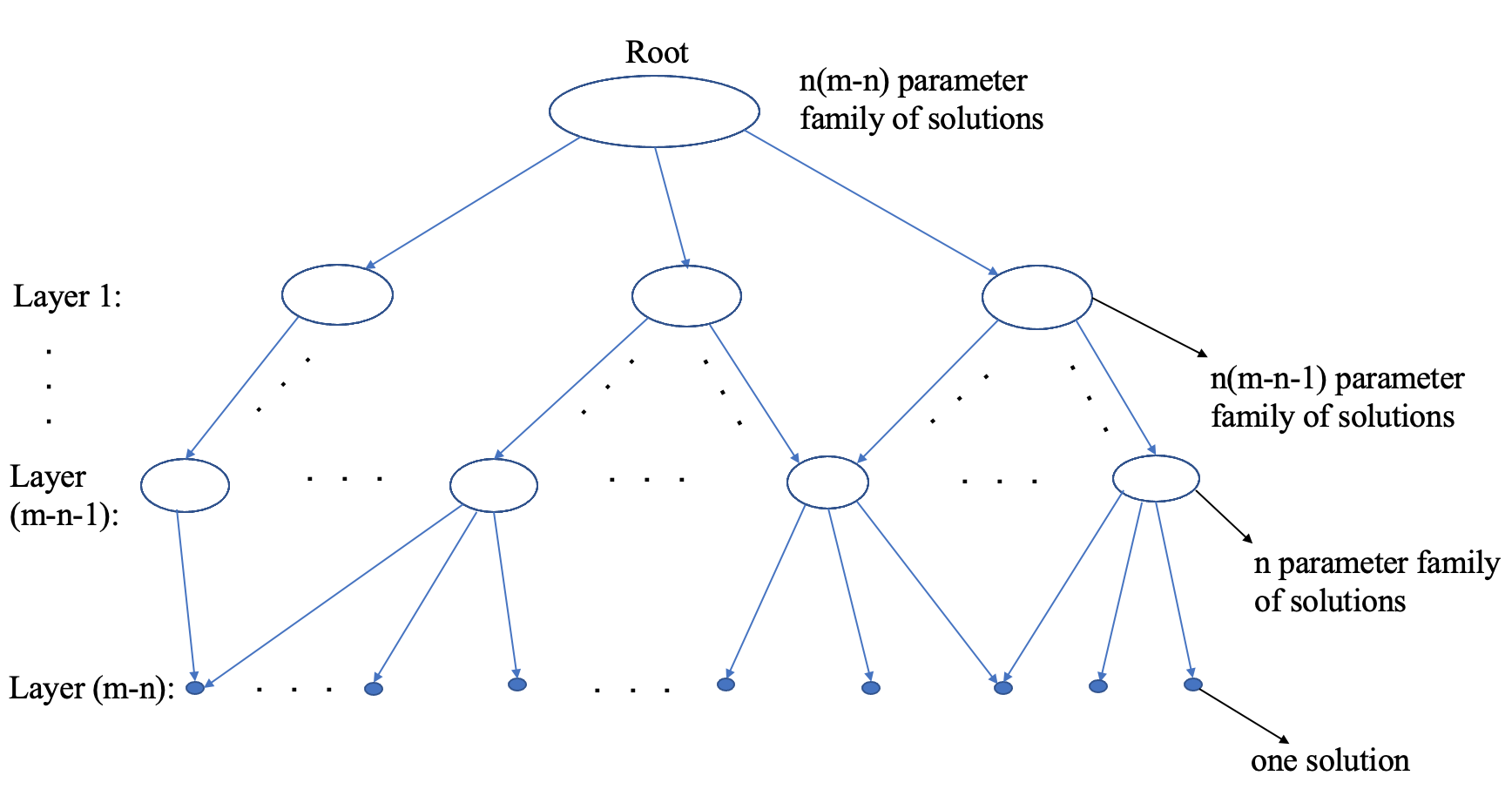}
\end{center}
\caption{Solutions to $\hat {\rm B}(\hat A) = A$ and the lattice of channel dropouts.  At the root, there is an $n(m-n)$-parameter family of solutions, and if any channel drops out, the solutions under the corresponding projection $P_I$ are an $n(m-n-1)$-parameter family.  At the bottom leaf nodes of the lattice where $m-n$ channels have become unavailable, there will typically be a single invariant $\hat A$.  Note that there are $m(m-1)\cdots(m-n+1)$ distinct paths from the root to the ${m\choose n}$ leaves.}
\label{fig:zs:Lattice}
\end{figure}

\begin{example}
We end this section with an example of random unavailability  of input channels.  Considering still the system of Example 1, define the feedback gain $K=\alpha B^T-\hat A$ where
\[
\hat A=\left(\begin{array}{cc}
0 & 0\\
0 & 1\\
0 & 0\end {array}
\right),
\]
and $\alpha =2$.  Consider channel intermittency in which each of the three channels is randomly available according to a two-state Markov process defined by
\[
\left(\begin{array}{c}
\dot p_u(t) \\
\dot p_a(t)
\end{array}\right) =
\left(\begin{array}{cc}
-\delta & \epsilon\\
\delta & -\epsilon\end{array} \right)\left(\begin{array}{c}
p_u(t) \\
p_a(t)
\end{array} \right),
\]
where $p_u\ =$ probability of channel unavailability; $p_a = 1-p_u\ =$ probability of the channel being available.  Assuming the availabilities of the channels are {\em i.i.d.} with this availability law, simulations show that for a range of parameters $\delta,\epsilon>0$ that
\begin{equation}
\dot x=(A+BP(t)K)\,x
\label{eq:jb:time-varying}
\end{equation}
is asymptotically stable with the time-dependent projection $P(t)$  being a diagonal matrix whose diagonal entries can take any combination of values 0 and 1.  The special case $\delta=3,\ \epsilon=3$ is illustrated in Fig.\ \ref{fig:zs:time-varying}.  The hypotheses of Theorem \ref{thm:jb:theorem2} characterize a worst case scenario in terms of stability since there may be time intervals over which (\ref{eq:jb:time-varying}) does not satisfy the hypotheses because $P(t)$ does not leave $\hat A$ invariant.
\end{example}

\begin{figure}[htbp]
\centerline {\includegraphics[scale=0.5]{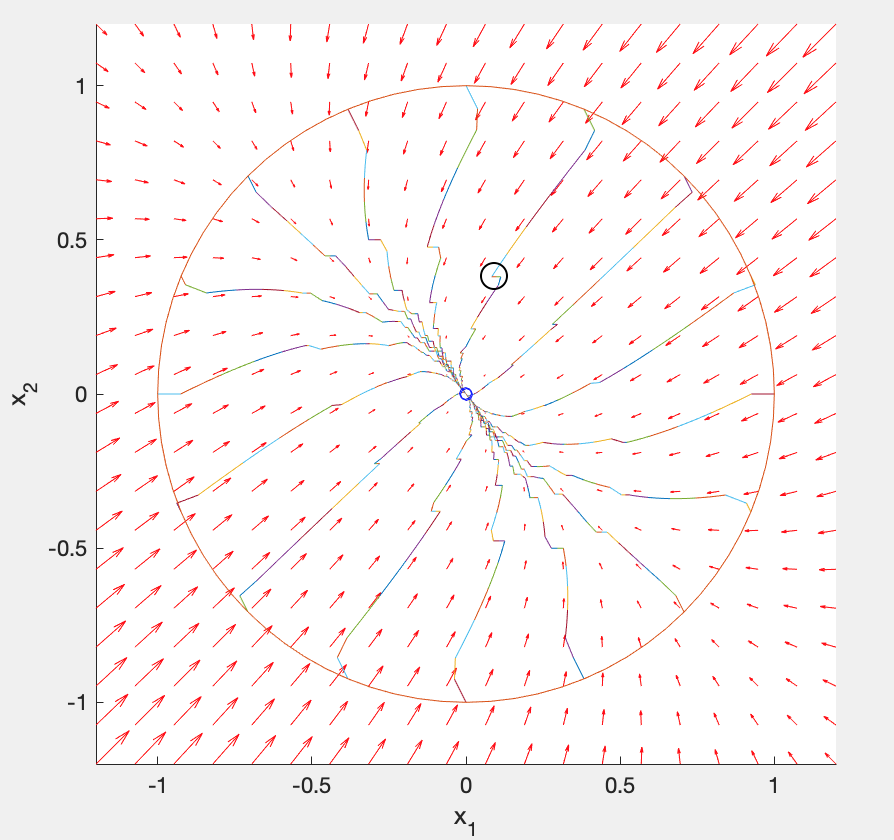}}
\caption{An example shows that when channels are intermittently available, the designed feedback can still maintain the asymptotic  stability of the system, even if there are some short short periods of instability as exhibited in the black circle.}
\label{fig:zs:time-varying}
\end{figure}

\bigskip
\section{Uncertainty reduction from adding a channel}
Suppose that instead of eq.\ (\ref{OffSet2}), uncertainty in each input channel is taken into account and the actual system dynamics are governed by
\begin{equation}
(\hat A + K)x_g+v+n_{\epsilon}=0,
\end{equation}
where $n_{\epsilon}$ is a Gaussian random vector whose entries are $i.i.d \in N(0,1)$.  Assume further that these channel-wise uncertainties are mutually  independent.  Then the asymptotic steady state limit $x_{\infty}$ will satisfy
\begin{equation}
E[\Vert x_{\infty}-x_g\Vert^2] = tr[(A+BK)^{-1}B\Sigma_{\epsilon}B^T((A+BK)^T)^{-1}],
\label{eq:jb:errorCov}
\end{equation}
where $\Sigma_{\epsilon}=I$ is the $m\times m$ covariance of the channel perturbations.  The question of how steady state error is affected by the addition of a noisy channel  is partially addressed as follows.
\begin{theorem}
Suppose the system (\ref{eq:jb:linear}) is controllable, and that $K$ and $\hat A$ have been chosen as in Theorem \ref{thm:jb:theorem2} so that $A+BK$ is Hurwitz and the control input $u(t) = \ Kx(t)+v$ has been chosen to steer the state to $x_g$ using $v$ given by (\ref{OffSet2}).  Let $ b\in\mathbb{R}^n$ and consider the augmented $n\times(m+1)$ matrix $\bar B=(B\vdots b)$.  Then with the control input $\bar u = \bar K x + \bar v$, where $\bar K=\left(\begin{array}{cc}K\\ k\end{array}\right)$, $k=-\alpha b^T $, and offset $\bar v= \left(\begin{array}{cc}v\\ 0\end{array}\right)$, the system $\dot x = Ax + \bar B\bar u$ is steered such that the steady state error covariance is
\[
\Sigma_{\bar B} = M(B \vdots b)\Sigma_{\epsilon}(B \vdots  b)^TM^T,
\]
where
\[
M=\left(A+(B\vdots b)(\begin{array}{c}K\\ k \end{array})\right)^{-1}.
\]
The corresponding steady state error covariance for $\dot x = Ax + Bu$ is
\[
\Sigma_B =(A+BK)^{-1}B\Sigma_{\epsilon}B^T((A+BK)^T)^{-1},\ and
\]
if the mean squared asymptotic errors under the two control laws are denoted by $err_{\bar B}$ and $err_B$,  then $err_{\bar B}<err_B$.
\label{thm:jb:uncertainty}
\end{theorem}
\begin{proof}
Let $M_B=(A+BK)^{-1}$ and $M_{\bar B}=\left(A+(B\vdots b)(\begin{array}{c}K\\ k \end{array})\right)^{-1}.$  We have $A+(B\vdots b)\left(\begin{array}{c} K\\ k\end{array}\right) = A+BK+ b k = A+B(-\alpha B^T-\hat A)-\alpha  b b^T = -\alpha (BB^T+bb^T$).  Hence $M_{\bar B}=-(1/\alpha)(BB^T+bb^T)^{-1}$, whereas $M_B=-(1/\alpha)(BB^T)^{-1}$.

Whence, in comparing mean square errors,
\[ 
\begin{array}{ccl}
err_{\hat B} & = & tr\left\{M_{\bar B} \bar B {\bar B}^TM_{\bar B}^T \right\}\\[0.07in]
& = & (1/\alpha^2)tr\left\{ (BB^T+ bb^T)^{-1}(BB^T+ bb^T)\right.\\
&&\ \ \ \ \ \ \ \ \ \ \left.(BB^T+ bb^T)^{-1}\right\}\\[0.07in]
& = & (1/\alpha^2)tr\left\{(BB^T+ bb^T)^{-1}\right\};
\end{array}
\]
\[
\begin{array}{ccl}
err_B & = & tr\left\{M_B BB^TM_B^T \right\}\\[0.07in]
& = & (1/\alpha^2)tr\left\{ ((BB^T)^{-1}(BB^T)(BB^T)^{-1}\right\}\\[0.07in]
& = & (1/\alpha^2)tr\left\{ (BB^T)^{-1}\right\}.
\end{array}
\]
Since $BB^T+bb^T > BB^T$ in the natural ordering of symmetric positive definite matrices, the conclusion of the theorem follows.
\end{proof}


\section{The nuanced control authority of highly parallel quantized actuation\label{sec:jb:quantized}}

Technological advances have made it both possible and desirable to re-examine classical linear feedback designs using a new lens as described in the previous sections.  In this section, we revisit the concepts of the previous sections in terms of quantized control along the lines of \cite{Baillieul2007},  \cite{Baillieul2004}, \cite{Baillieul2004a}, \cite{Baillieul2004b}, and \cite{Nair2}.  The advantages of large numbers of input channels in terms of reducing energy (\cite{Baillieul2}) and uncertainty (Theorem \ref{thm:jb:uncertainty} above) come at a cost that is not easily modeled in the context of linear time-invariant (LTI) systems.  Indeed as the number of control input channels increases, so does the amount of information that must be processed to implement a feedback design of the form of Theorem \ref{thm:jb:theorem2} in terms of the {\em attention} needed (measured in bits per second) to ensure stable motions.  (See \cite{Baillieul2002} for a discussion of attention in feedback control.)  An equally subtle problem with the models of massively parallel input described in the preceding sections is the need for asymptotic scaling of inputs as the number of channels grows.  If the number of input channels is large enough, then the feedback control $u=Kx$ for $K=-B^T$ will be stabilizing for any system (1).  This is certainly plausible based on Gershgorin's Theorem and the observation that if $B_1$ has full rank $m$, and $B_2$ is obtained from $B_1$ by adding one (or more) columns, then $B_2B_2^T > B_1B_1^T$ in terms of the standard order relationship on positive definite matrices.  A precise statement of how fast the matrix $BB^T$ grows from adding columns in given by the following.
 
\begin{proposition}
Let $B$ be an $2\times m$ matrix with $m>2$ whose columns are unit vectors uniformly spaced on the unit circle $S^1$.  Then, the spectral radius of $BB^T$ is $m/2$.
\label{prop:jb:one}
\end{proposition}
\begin{proof}
The proof is by direct calculation.  There is no loss of generality in assuming the $m$ uniformly distributed vectors are $(\cos \frac{2 k \pi}{m},\sin \frac{2 k \pi}{m})$.  $BB^T$ is  then given explicitly by
\[
BB^T=\left(\begin{array}{ll}
\sum_{k=1}^m \cos^2 \frac{2 k \pi}{m} & \sum_{k=1}^m \sin  \frac{2 k \pi}{m} \cos  \frac{2 k \pi}{m}\\[0.05in]
 \sum_{k=1}^m \sin  \frac{2 k \pi}{m} \cos  \frac{2 k \pi}{m} & \sum_{k=1}^m \sin^2 \frac{2 k \pi}{m} 
 \end{array}
 \right).
\]
Standard trig identities show the off-diagonal matrix entries are zero and the diagonal entries are $m/2$, proving the proposition.
\end{proof}
For $n>2$, the conclusion is similar but slightly more complex.  We consider $n\times m$ matrices $B$ ($m>n$) whose columns are unit vectors in $\mathbb{R}^n$.  Following standard constructions of Euler angles (e.g. http://www.baillieul.org/Robotics/740\_L9.pdf), we define spherical coordinates (or generalized Euler angles) on the $n-1$-sphere $S^{n-1}$ whereby points $(x_1,\dots,x_n)\in S^{n-1}$ are given by
\[
\begin{array}{lll}
x_1& =& \sin\theta_1\sin\theta_2\cdots\ \sin\theta_{n-1}\\ 
x_2& = & \sin\theta_1\sin\theta_2\cdots\ \cos\theta_{n-1}\\
x_3& =&\sin\theta_1\sin\theta_2\cdots\cos\theta_{n-2}\\
&\vdots &\\
x_{n-1}  & = & \sin\theta_1\cos\theta_2\\
x_n &=& \cos\theta_1,
\end{array}
\]
where $0\le\theta_1\le\pi,\ 0\le\theta_j<2 \pi$ for $j=2,\dots,n-1$.  We shall call a distribution of points on $S^{n-1}$ {\em parametrically regular} if it is given by $\theta_1 = \frac{j\pi}{N_1},\ (j=0,\dots,N_1)$, and $\theta_k = \frac{2j\pi}{N_k},\ (j=1,\dots,N_k)$ for $2\le k\le n-1$.  The following extends Proposition \ref{prop:jb:one} to $n\times m$ matrices where $n>2$.

\begin{theorem}
Let the $n\times m$ matrix $B$ comprise columns consisting of all unit $n$-vectors associated with the parametrically regular distribution $(N_1,\dots,N_{n-1})$ with all $N_j>2$.  $B$ then has $m=(N_1+1)N_2\cdots N_{n-1}$ columns; $BB^T$ is diagonal, and the largest diagonal entry (eigenvalue) is $\frac{N_1+2}{2}N_2\cdots N_{n-1}$.
\end{theorem}
\begin{proof}
The matrices $B$ of the theorem have the form
\[
\left(\begin{array}{ccc}
& \sin(\frac{j_1\pi}{N_1})\sin(\frac{j_22\pi}{N_2})\cdots\sin(\frac{j_{n-1}2\pi}{N_{n-1}}) & \\
\cdots & \sin(\frac{j_1\pi}{N_1})\sin(\frac{j_22\pi}{N_2})\cdots\cos(\frac{j_{n-1}2\pi}{N_{n-1}}) & \cdots \\
\cdots & \vdots & \cdots \\
& \cos(\frac{j_1\pi}{N_1}) & \end{array}\right).
\]
The entries in the product $BB^T$ may be factored into products of sums of the form
\[
\begin{array}{cccc}
\ \sum_{k=1}^{N_1}\cos^2(\frac{k\pi}{N_1}) = \frac{N_1+2}{2}, &\sum_{k=1}^{N_j}\sin^2(\frac{k2\pi}{N_j}) = \frac{N_j}{2},\\ \sum_{k=1}^{N_j}\cos^2(\frac{k2\pi}{N_j}) = \frac{N_j}{2},&\ \ \ \ \ \ \  \sum_{k=1}^{N_j}\sin(\frac{k2\pi}{N_j})\cos(\frac{k2\pi}{N_j}) = 0,  (j=2,\dots,n-1),&\end{array}
\]
and from this the result follows.
\end{proof}
\begin{remark}\rm
Simulation experiments show that for matrices $B$ comprised of columns of random unit vectors in $\mathbb{R}^n$ that are approximately parametrically regular in their distribution, the matrix norm (=largest singular value) of $BB^T$ is ${\cal O}(m)$, in agreement with the theorem.
\end{remark}

To keep the focus on channel intermittency and pursue the {\em emulation} problems of the following section, we shall assume the parameter $\alpha$ appearing in Theorem \ref{thm:jb:theorem2} is inversely related to the size, $m$, of the matrix $B$, and for the case of quantized inputs considered next, other bounds on input magnitude may apply as well.

We conclude by considering discrete-time descriptions of (\ref{eq:jb:linear}) in which the control inputs to each channel are selected from finite sets having two or more elements.  To start, suppose (\ref{eq:jb:linear}) is steered by inputs that take continuous values in $\mathbb{R}^m$ but are piecewise constant between uniformly spaced instants of time: $t_0<t_1< \cdots;\ t_{k+1}-t_k = h$, and $u(t)=u(t_k)$, for $t_k\le t <t_{k+1}$.  Then the state transition between sampling instants is given by
\begin{equation}
x(t_{k+1})=F_hx(t_k) + \Gamma_h u(t_k)
\label{eq:jb:sampled}
\end{equation}
where
\[
F_h=e^{Ah}\ \ {\rm and}\ \ \Gamma_h=\int_0^he^{A(h-s)}\,ds\cdot B.
\]
When $B$ (and hence $\Gamma_h $) has more columns than rows, there is a parameter lifting $\hat\Gamma:\mathbb{R}^{m\times n}\rightarrow\mathbb{R}^{n\times n}$, and it is possible to find $\hat F$ satisfying $\hat\Gamma(\hat F)=F_h$.  The relationship between $\hat B$ and $\hat\Gamma$ and between the lifted parameters $\hat A$ and $\hat F$ is highly nonlinear, and will not be explored further here.  Rather, we consider first order approximations to $F_h$ and $\Gamma_h$ and rewrite (\ref{eq:jb:sampled}) as
\[
\begin{array}{ccl}
x(t_{k+1})&=&(I+Ah)x(t_k) + hBu(t_k) + o(h)\\[0.05in]
&=&x(t_k)+h(Ax(t_k) + Bu(t_k)) + o(h).
\end{array}
\]
We may use lifted parameters of the first order terms and write the first order, discrete time approximation to (\ref{eq:jb:linear}) as
\begin{equation}
x(k+1)=(I+h \hat{\rm B}(\hat A))x(k) + hBu(k).
\label{eq:jb:discrete}
\end{equation}

Having approximated the system in this way, we consider the problem of steering (\ref{eq:jb:discrete}) by input vectors $u(k)$ having each entry selected from a finite set, say $\{-1,1\}$ in the simplest case.  A variety of MEMS devices that operate utilizing very large arrays of binary actuators can be modeled in this way, and successful control designs in adaptive optics applications have been either open-loop or hybrid open-loop combined with modulation from optical wavefront sensors (\cite{Baillieul1999},\cite{Bifano}).  The approach being pursued in the present work aims to close feedback loops using real-time measurements of the state.  For any state $x(k)$, binary choices must be made to determine the value of the input to each of the $m$ channels.  Control  design is thus equivalent to implementing a selection function along the lines of \cite{Baillieul2002}.  Some details of research on such designs are described in the next section, but complete results will appear elsewhere.

\section{Conclusions and work in progress}

It is easy to show that if inputs $u(k)$ to (\ref{eq:jb:discrete}) can take a continuum of values in an open neighborhood of the origin in $\mathbb{R}^m$, then feedback laws based on sample-and-hold versions of (\ref{eq:jb:defK}) can be designed to asymptotically steer the state of (\ref{eq:jb:discrete}) to the origin of the state space $\mathbb{R}^n$.  Current work is aimed at extending the ideas we have described in the above sections of the paper to a theory of feedback control designs for (\ref{eq:jb:discrete}) in which control inputs at each time step are selected from a finite set ${\cal U}$ of admissible control values.  One goal is design of both selection functions and modulation strategies whereby systems of the form (\ref{eq:jb:discrete}) with finitely quantized feedback can emulate the motions of continuous time systems like the ones treated in the preceding sections.  In the quest to find control theoretic abstractions of widely studied problems in neuroscience, the work is organized around three themes: {\em neural coding} (characterizing the brain responses to stimuli by collective electrical activity---spiking and bursting---of groups of neurons and the relationships among the electrical activities of the neurons in the group), {\em neural modulation} (real-time transformations of neural circuits induced by various means including chemically or through neural input from projection neurons) wherein connections within the anatomical connectome are made to specify the functional brain circuits that give rise to behavior \cite{Marder}, and {\em neural emulation} (finding patterns of neural activity that associate potential actions to the outcomes those actions are expected to produce, \cite{Colder}).

The simplest coding and emulation problems in this context involve finding state-dependent binary vectors $u[k]\in \{-1,1\}^m$ that elicit the desired state transitions of (\ref{eq:jb:discrete}).  For finite dimensional linear systems, it is natural to consider using  inputs that are both spatially and temporally discrete to emulate continuous vector fields as follows.
Let $H$ be an $n\times n$ Hurwitz matrix.  A discretized Euler approximation to the solution of $\dot x = Hx;\ x(0)=x_0$ is
\[
x(k+1)=(I+hH)x(k); x(0)=x_0.
\]
Two specific problems of emulating this system by (\ref{eq:jb:discrete}) with binary inputs are:
\begin{itemize}
\item Find a partition of the state space $\{U_i\,:\, \cup\  U_i = \mathbb{R}^n;\ \ U_i^o\cap U_j^o=\emptyset;\ U_i^o={\rm interior}\ U_i\}$ and a rule for assigning coordinates of $u(k)$ to be either +1 or -1, so that for each $x\in U_i$, $Ax + Bu(k)$ is as close as possible to $Hx$ (in an appropriate metric).
\item Find a partition and a rule that for each $U_i$ both assigns $\pm 1$ coordinate entries of $u(k)$ together with a coordinate projection operation $P(k)$ determining which channels are operative in the partition cell $U_i$ with the property the $Ax+BP(k)u(k)$ is as close as possible to $Hx$ (in an appropriate metric).
\end{itemize}
\begin{figure}[h]
\begin{center}
\includegraphics[scale=0.3]{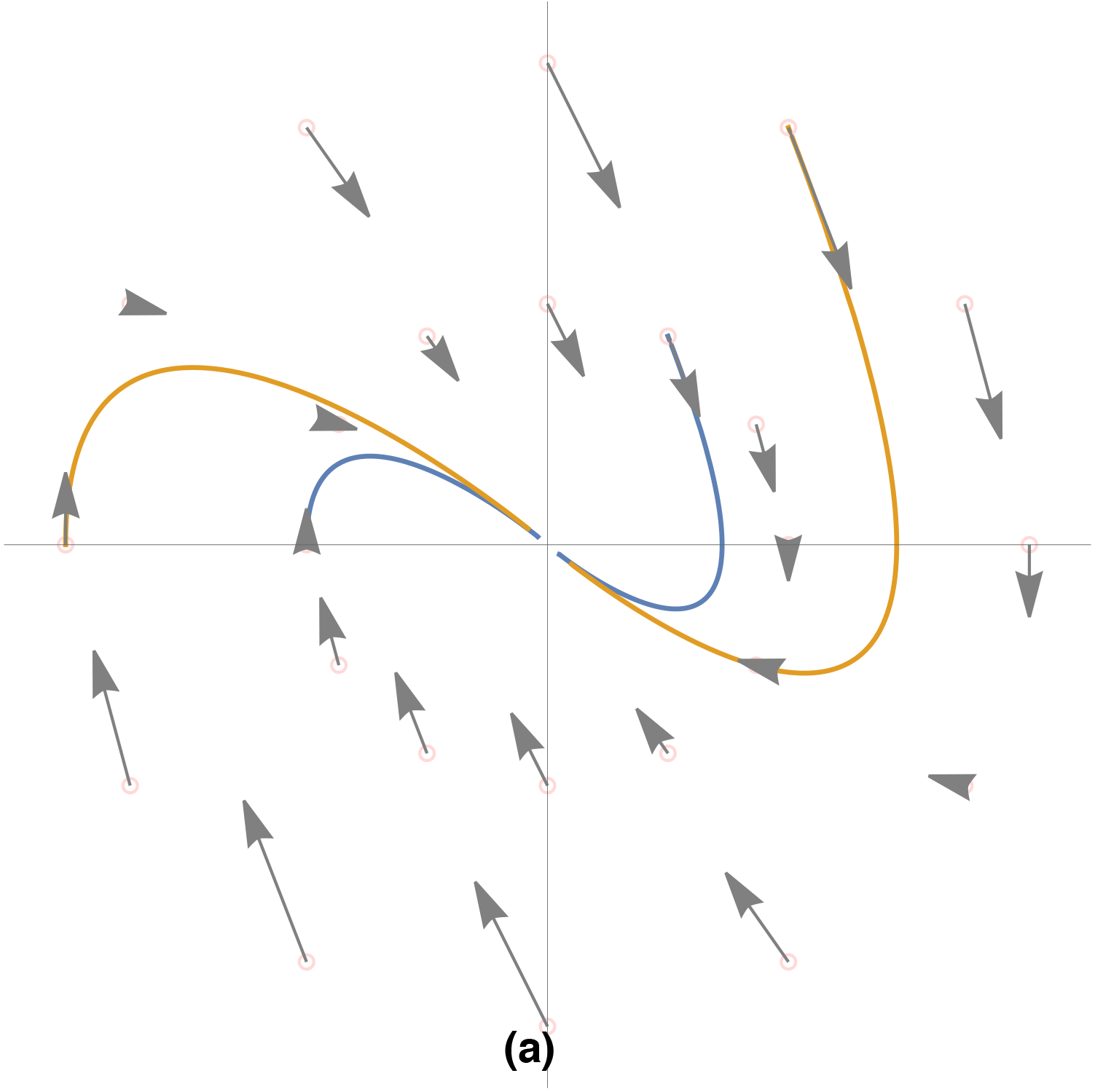}
\includegraphics[scale=0.3]{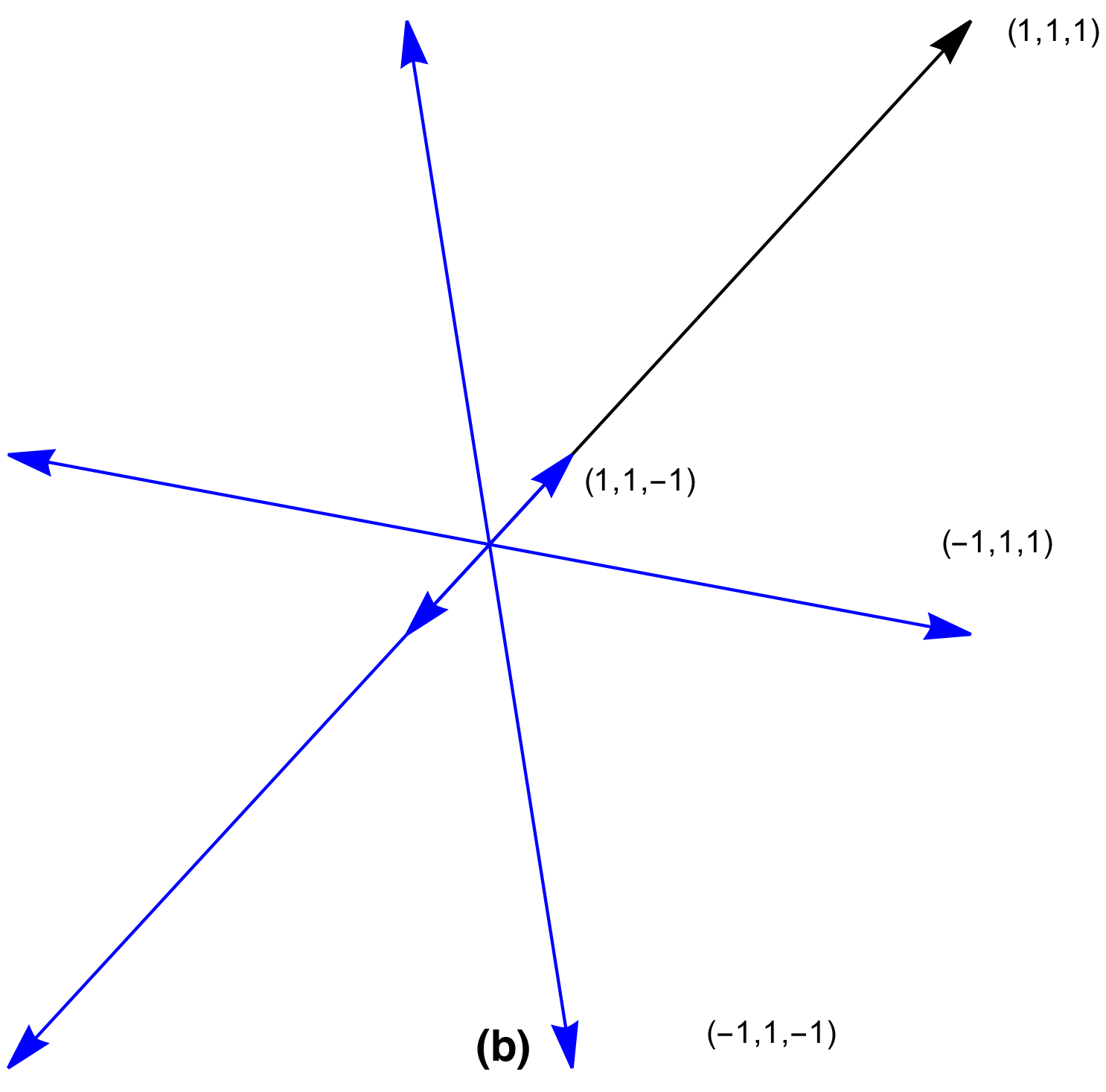}\\[0.1in]
\includegraphics[scale=0.3]{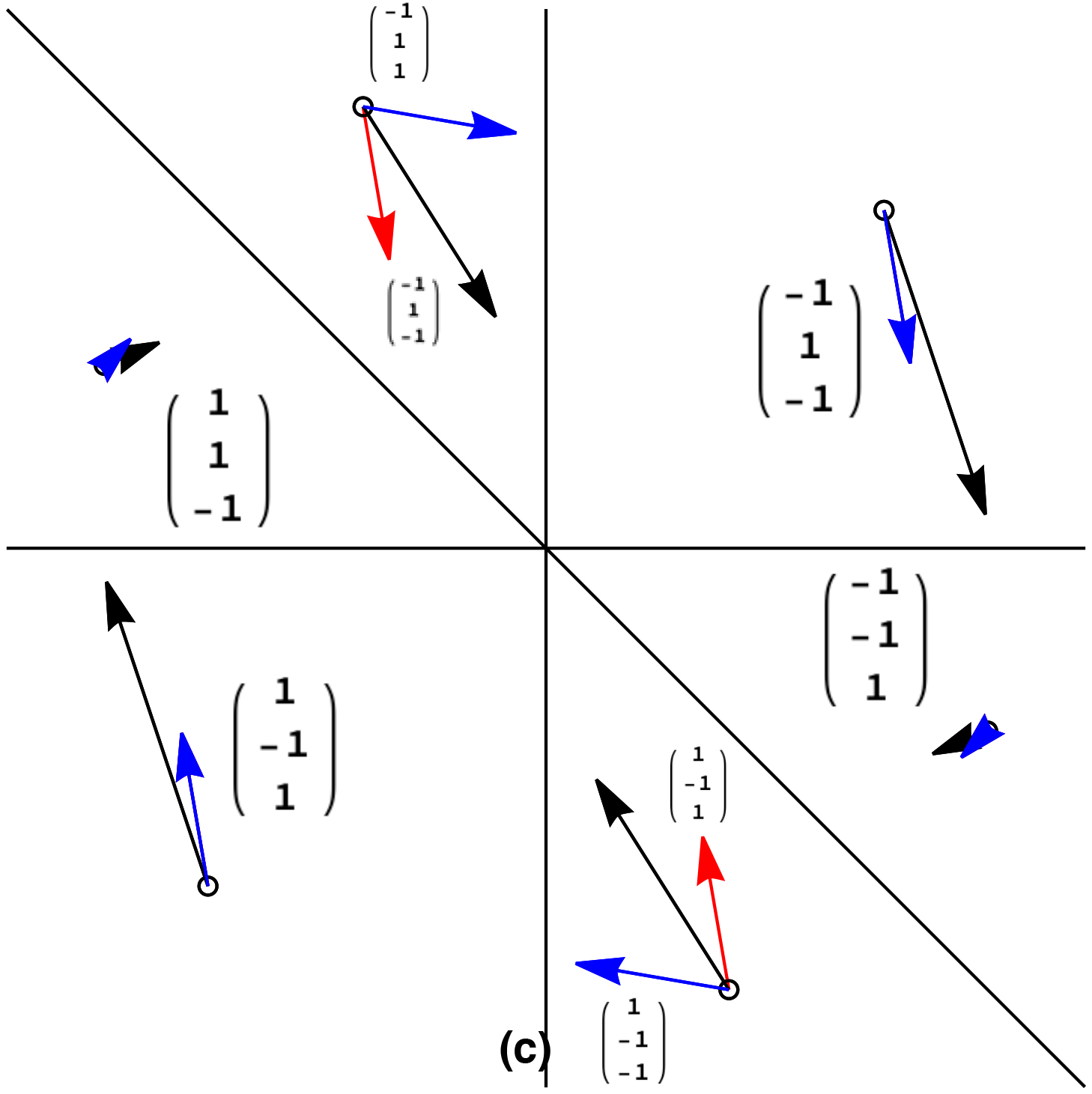}
\includegraphics[scale=0.3]{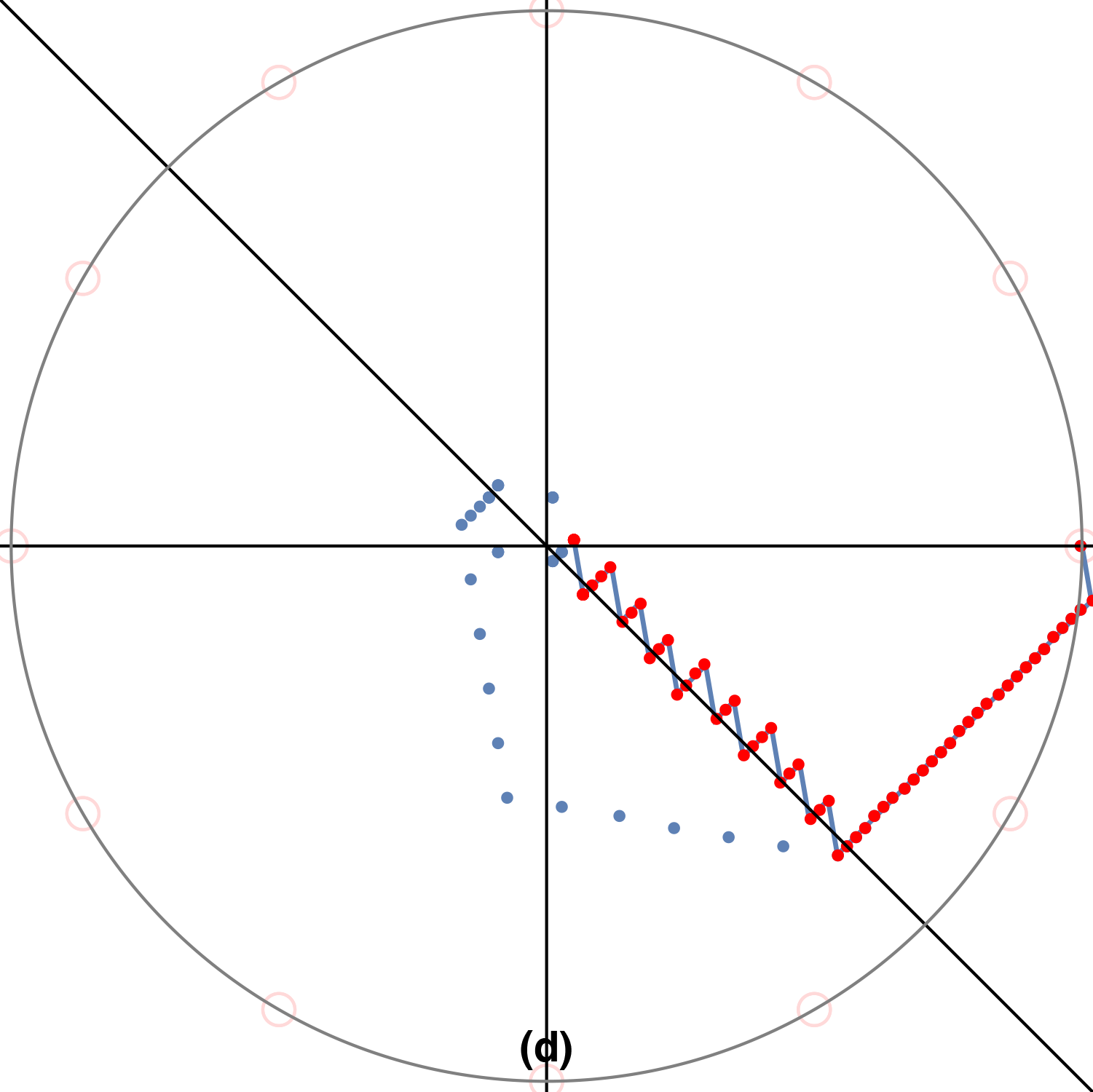}
\end{center}
\caption{(a) Linear vectors fields scale in magnitude with the distance from the origin. (b) The choices of vector directions for the matrix $B$ with binary $\pm 1$ inputs.  (c) Choices of vector directions as shown in (b) and specified by $\pm 1$ inputs to the channels of $B$ to provide plausible approximations to the continuous vector field of (a), and (d) is typical paths followed by applying the $\pm 1$ inputs as prescribed in (c).}
\label{fig:jb:linear}
\end{figure}
\begin{example}
Returning to the two-state, three-input examples considered in Section II, suppose $A=0$ and $B=\left(\begin{array}{ccc}
0&1&1/\sqrt{2}\\
1&0&1/\sqrt{2}\end{array}\right)$.  We have normalized the third column of $B$ in order that all three channels are equally influential in steering with binary inputs chosen from $\{-1,1\}$.  Consider $H=\left(\begin{array}{cc} 0&1\\-1&-2\end{array}\right)$, a Hurwitz matrix with both eigenvalues equal to -1.  A sketch of the vector field and flow is shown in Fig.~\ref{fig:jb:linear}(a).  It is noted that length scales of linear vector fields increase in proportion to their distance from the origin.  This needs to be accounted for in interpolations using binary affine maps in a way that may be similar to the way that neurons in the entorhinal cortex---grid cells in particular---encode information about length scales.  In mammalian brains, grid cells show stable firing patterns that are correlated with hexagonal grids of locations that an animal visits as it moves around a range space.   Grid cells appear to be arranged in discrete modules anatomically arranged in order of increasing scale  along the dorso-ventral axis of the medial entorhinal cortex \cite{Stensola, Bush}, with each module's neurons firing at module-specific scales.
\end{example}

The mechanisms of neural coding and modulation are ares of active research in brain science, and it is in part for this reason that we have avoided being precise in specifying approximation metrics in this emulation problem.  Both binary codings depicted in Fig.~\ref{fig:jb:linear} steer the quantized system (\ref{eq:jb:discrete}) toward the origin---but along qualitatively different paths.  In observing animal movements in field and laboratory settings, it is found that they can exhibit a good deal of individuality in choosing paths from a common start point to a common goal (\cite{KongEtAl}).  There is also evidence that among regions of the brain that guide movement, some exhibit neural activity that is more varied than others such as the grid cell modules in the entorhinal cortex \cite{LeeCubed}, which tend to be similar in neuroanatomy from one animal to the next.  As future work will be focused on systems with large numbers of simple inputs along with learning strategies, we will be aiming to understand the emergence of multiple and diverse solutions that meet similar control objectives.

In treating neuro-inspired approaches to input modulation along the lines suggested by the second bullet above, we note that with enough control channels available,  it is possible to pursue quantized control input designs in which a large enough and appropriately chosen group of binary inputs can simulate certain amplitude-modulated analog signals.  How this approach scales with control task complexity and numbers of neurons needed for satisfactory execution remains ongoing work.

Finally, we note that our work exploring neuromimetic control has focused on linear systems only because they are the most widely studied and best understood.  There is no reason that exploration of nonlinear control systems that model land and air vehicles cannot be approached in ways that are similar to what we have presented.  For such models, connections between task geometry and time-to-execute can be studied.  Deeper connections to neurobiological aspects of sensing and control in and of themselves will also call for nonlinear models.  This is foreshadowed in the piecewise selection functions that define the quantized control actions in this section.  Connections with established theories of threshold-linear networks (\cite{Curto}) and with other models competitive neural dynamics (\cite{Layton}) remain under investigation.  The neurobiology literature on the control of physical movement is very rich, and there appears to be much to explore.

\bibliographystyle{IEEEtran}

\end{document}